\documentclass[preprint,review, 12pt]{elsarticle}

\usepackage{graphicx}
\usepackage{amssymb}
\usepackage{amsthm}
\usepackage{amsmath}
\newtheorem{theorem}{Theorem}
\theoremstyle{definition}
\newtheorem{definition}{Definition}
\theoremstyle{remark}
\newtheorem{note}{Note}
\newdefinition{remark}{Remark}

\journal{ArXiv}

\begin{document}

\begin{frontmatter}

%% Title, authors and addresses

%% use the tnoteref command within \title for footnotes;
%% use the tnotetext command for the associated footnote;
%% use the fnref command within \author or \address for footnotes;
%% use the fntext command for the associated footnote;
%% use the corref command within \author for corresponding author footnotes;
%% use the cortext command for the associated footnote;
%% use the ead command for the email address,
%% and the form \ead[url] for the home page:
%%
%% \title{Title\tnoteref{label1}}
%% \tnotetext[label1]{}
%% \author{Name\corref{cor1}\fnref{label2}}
%% \ead{email address}
%% \ead[url]{home page}
%% \fntext[label2]{}
%% \cortext[cor1]{}
%% \address{Address\fnref{label3}}
%% \fntext[label3]{}

\title{Comfortability of a Team in  Social Networks}

\author[NITT]{Lakshmi Prabha S}
\ead{jaislp111@gmail.com}

\author[NITT]{T.N.Janakiraman\corref{cor1}}
\ead{janaki@nitt.edu}

\cortext[cor1]{Corresponding Author - T.N.Janakiraman.}

\address[NITT]{Department of Mathematics, National Institute of Technology, 
Trichy-620015, Tamil Nadu, India.} 

%% use optional labels to link authors explicitly to addresses:
%\author[label1,label2]{<author name>}
%% \address[label1]{<address>}
%% \address[label2]{<address>}
%\markright{Approximation Algorithm for TDS-Max WE}        

%% use optional labels to link authors explicitly to addresses:

\begin{abstract}
There are many indexes (measures or metrics) in Social Network Analysis (SNA), like density, cohesion, etc. We have defined a new SNA index called \lq \lq comfortability \rq\rq. In this paper, core comfortable team of a  social network is defined based on graph theoretic concepts and some of their structural properties are analyzed. 
 Comfortability is one of the important attributes  (characteristics) for a successful team work. So, it is necessary  to find a comfortable and successful team in any given social network.   

It is proved that forming core comfortable team in any  network is NP-Complete using the concepts of domination in graph theory. Next, we give two polynomial-time approximation algorithms for finding such a core comfortable team in any given network with  performance ratio $O(\ln \Delta)$, where $\Delta$ is the maximum degree of a given network (graph). The time complexity of the algorithm is proved to be $O(n^{3})$, where $n$ is the number of persons (vertices) in the network (graph). It is also proved that the algorithms give good results in scale-free networks.

\end{abstract}


\end{frontmatter}

%%
%% Start line numbering here if you want
%%
% \linenumbers

\section{Introduction}
\label{intro}
There are many factors, lack of  which affect the group or team effectiveness. Team processes describe subtle aspects of interaction and patterns of organizing, that transform input into output. 
The team processes will be described in terms of seven characteristics: coordination, communication,
cohesion, decision making, conflict management, social relationships and performance feedback. The readers are directed to refer Michan et al.~\cite{team} for further details of characteristics of team and Forsyth~\cite{group} for more details on group dynamics. In this paper, we discuss about an attribute or characteristic   called \lq\lq COMFORTABILITY \rq\rq, which is also essential for a successful team work. So, we defined it as a new SNA index in our paper~\cite{JRLP}. 

Since the beginning of Social Network Analysis, Graph Theory has been a very important tool both to represent social structure and to calculate some indexes, which are useful to understand several aspects of the social context under analysis. Some of the existing indexes (measures or metrics) are betweenness, bridge, centrality, flow betweenness centrality, centralization, closeness, clustering coefficient, cohesion, degree,   density,  eigenvector centrality,  path length. Readers are directed to refer Martino et. al.~\cite{diameter} for more details on indexes in SNA.  In our paper~\cite{JRLP}, we defined a new  SNA index called \lq comfortability \rq. Based on this index, we have  defined comfortable team, better comfortable team and highly comfortable team in our paper~\cite{JRLP} and totally comfortable team in our paper~\cite{china}.

Let the social network be represented in terms of a graph, with the vertex of the graph denotes a person (an actor) in the social network and an edge between two vertices in a graph represents relationship between two persons in the social  network. All the networks  are connected networks in this paper, unless otherwise specified. If the given network is disconnected, then each connected component of the network can be considered and hence it is enough to consider only connected networks.  Hereafter, the word \lq team\rq \ represents induced sub network (sub graph) of a given network (graph). 

Following are some introduction for \textbf{basic graph theoretic  concepts}.
Some basic definitions from Slater et al.~\cite{Slater} are given below.

The  graphs  considered  in  this  paper  are  finite, simple, connected  and  undirected, unless otherwise specified. For a graph $G$, let $V(G)$ (or simply $V$) and $E(G)$ denote its vertex (node) set and edge set respectively and $n$ and $m$ denote the cardinality of those sets respectively. The \textit{degree} of a vertex $v$ in a graph $G$ is denoted by $deg_{G}(v)$. The \textit{maximum degree} of the graph $G$ is denoted by $\Delta(G)$. The length of any shortest path between any two vertices $u$ and $v$ of a connected graph $G$ is called the \textit{distance} between $u$ and $v$ and is denoted by $d_{G}(u,v)$. For a connected graph  $G$, the \textit{eccentricity} $e_{G}(v) = \max\{d_{G}(u,v): u\in V(G)\}$. If there is no confusion, we simply use the notions  $deg(v)$, $d(u,v)$ and $e(v)$ to denote degree, distance and eccentricity respectively for the concerned graph. The minimum and maximum eccentricities are the \textit{radius} and  \textit{diameter} of $G$, denoted by  $r(G)$ and $diam(G)$ respectively. A vertex with  eccentricity $r(G)$ is called a \textit{central vertex} and a vertex with eccentricity $diam(G)$ is called a \textit{peripheral vertex}. 
   
   For  $v \in V(G)$, \textit{neighbors} of $v$ are the vertices adjacent to $v$ in $G$. The neighborhood $N_{G}(v)$ of $v$ is the set of all neighbors of  $v$ in $G$. It is also denoted by $N_{1}(v)$. $N_{j}(v)$ is the set of all vertices at distance $j$ from $v$ in $G$. A vertex $u$ is said to be an \textit{eccentric vertex} of $v$, when $d(u, v) = e(v)$. If $A$ and $B$ are not necessarily disjoint sets of vertices, we define
the distance from $A$ to $B$ as
$dist(A,B) = \ min\{ d(a, b) : a \in A, b \in B \}$. \textit{Cardinality} of a set $D$ represents the number of vertices in the set $D$. Cardinality of $D$ is denoted by $|D|$.

 A vertex of degree one is called a \textit{pendant vertex}. A \textit{walk} of length $j$ is an alternating sequence $W:  u_0, e_1, u_1, e_2, u_2,\ldots, u_{j-1}, e_j, u_j$ of vertices and edges with $e_i = u_{i-1}u_i$. If all $j$ edges are distinct, then $W$ is called a \textit{trail}. A walk with $j+1$ distinct vertices $u_0, u_1, \ldots, u_j$ is a \textit{path} and if $u_0 = u_j$ but $u_1, u_2, \ldots, u_j$ are distinct, then the trail is a \textit{cycle}. A path of length $n$ is denoted by $P_n$ and a cycle of length $n$ is denoted by $C_n$. A graph $G$ is said to be \textit{connected} if there is a path joining each pair of nodes. A component of a graph is a maximal connected sub graph. If a graph has only one component, then it is connected, otherwise it is \textit{disconnected}. A \textit{tree} is a connected graph with no cycles (acyclic).

  We say that   $H$ is a \textit{sub graph} of a graph $G$, denoted  by $H < G$, if  $V(H) \subseteq V(G)$  and $uv \in E(H)$ implies $uv \in E(G)$ . If a sub graph $H$ satisfies  the added property that for every pair $u,v$ of vertices, $uv \in E(H)$ if and only if $uv \in E(G)$, then $H$ is called an \textit{induced sub graph} of $G$. The induced sub graph $H$ of $G$ with $S = V(H)$ is called the sub graph induced by $S$ and is denoted by $\left\langle S|G\right\rangle$ or simply $\left\langle S\right\rangle$.
  
  Let $k$ be a positive integer. The $k^{th}$  \textit{power} $G^{k}$ of a graph $G$  has $V(G^{k}) = V(G)$ with $u,v$ adjacent in $G^{k}$ whenever $d(u,v) \leq k$.

  The concept of domination was introduced by Ore~\cite{Ore} . A set $D \subseteq V(G)$  is  called a \textit{dominating set} if every vertex  $v$ in $V$ is either an element of $D$ or is adjacent to an element of  $D$.  A  dominating set  $D$ is a minimal dominating set if  $D-\{v\}$  is  not a dominating set for any $v \in D$. The domination number  $\gamma(G)$ of a graph $G$ equals the minimum cardinality of a dominating set in $G$. 
  
  A set $D$ of vertices in a connected graph $G$ is called a \textit{$k$-dominating set} if every vertex in $V- D$ is within distance $k$ from some vertex of $D$. The concept of the $k$-dominating set was
introduced by Chang and Nemhauser~\cite{Chang,Chang2} and could find applications for many situations and structures which give rise to graphs; see the books by Slater et al~\cite{Slater, REF6}. So, dominating set is nothing but 1-dominating set.
  
  Sampath  Kumar  and  Walikar~\cite{CDS}  defined  a connected  dominating  set   $D$  to  be  a dominating set $D$,  whose  induced   sub-graph  $\left\langle D\right\rangle$ is \textit{connected}.  The  minimum  cardinality  of  a connected  dominating  set  is  the  connected domination number $\gamma_c(G)$. 
  
The readers are also directed to refer Slater et al.~\cite{Slater} for  further details of basic definitions, not given in this paper.

Let us recall the terminologies as follows: The symbol ($\rightarrow$)  denotes  \lq\lq represents \rq\rq 
\begin{itemize}
	\item Graph $\rightarrow$ Social Network (connected)
	\item Vertex of a graph $\rightarrow$ Person in a social network
	\item Edge between two vertices of a graph $\rightarrow$ Relationship between two persons in a social network
	\item Induced subgraph of a graph $\rightarrow$  \textbf{Team or Group} of a social network.
\end{itemize}
 
Given a connected network of people.  Our problem is to find a team (sub graph) which is less dispersive, highly flexible and performing better. 

\begin{note}
\textbf{Notation 1:}\\
In all the figures of this paper, 
\begin{itemize}
	\item $\{v_1, v_2,\ldots, v_n\}$ represent the vertex set of the graph $G$, that is,\\ $V(G) = \{v_1, v_2,\ldots, v_n\}$. 
	\item The numbers besides every vertex represents the eccentricity of that vertex. 
For example, in Figure~\ref{fig1}, in the graph $G$, $e(v_1) = 5$, $e(v_2) = 4$, $e(v_3) = 3$, $e(v_4) = 3$, $e(v_5) = 4$ and $e(v_6) = 5$. 
\item The set notation  $D = \{v_1, v_2, \ldots, v_n\}$ represents only the individual persons but does not represent the relationship between them.
\item The notation $\left\langle D \right\rangle$ represents the team. $\left\langle D \right\rangle$ is the induced sub graph of $G$, which represents the persons as well as the relationship between them.
So, the set $D$ represents only the team members and the team represents the persons with their relationship.
\end{itemize}
\end{note}
The remaining part of the paper is organized as follows:
\begin{itemize}
\item Section~\ref{prior} discusses about prior work.
	\item Section~\ref{gcomf} defines core comfortable team and analyses the concept with some examples.
\item In section~\ref{algo}, an approximation algorithm is given for finding core comfortable team in any given network,  with illustrations. Time complexity of the algorithm is analyzed. Also, correctness of the algorithm  and performance ratio of the algorithm are proved in this section.
\item In Section~\ref{second}, another approximation algorithm is given, based on connected dominating set, for finding core comfortable team in any given network,  with illustrations. Time complexity, correctness   and performance ratio of this algorithm are discussed in this section.
\item In section~\ref{adv}, some advantages of the two algorithms are given.
\item Section~\ref{conc} concludes the paper and discusses about some future work.
\end{itemize}

\section{Prior Work}
\label{prior}
In our paper ~\cite{JRLP}, we  defined characteristics of a good performing team and mathematically formulated them and given approximation algorithm for finding such a good performing team. \textit{In order to make this paper self-contained, we have given all the necessary definitions, examples and properties from our paper~\cite{JRLP}, which are needed for this paper}, in this section.
\begin{definition}~\cite{JRLP} 
\label{good} A team  is said to be  \textbf{good performing or successful} if the team is 
\begin{enumerate}
	\item less dispersive
	\item having good communication among the team members
	\item easily accessible to the non- team members 
	\item a good service provider to the non-team members (for the whole network).
\end{enumerate}
\end{definition}
\textbf{Mathematical Formulation}\\
\textbf{Domination $\rightarrow$ good service provider to the non-team members}.\\
\textbf{Connectedness $\rightarrow$ good communication among team members}.\\
 
\begin{definition}\textbf{Less Dispersive Set:}~\cite{JRLP}
\label{less}
A set $D$ is said to be less dispersive, if $e_{\left\langle D\right\rangle}(v) <  e_G(v)$, for every vertex $v \in D$.
\end{definition}
\begin{definition}\textbf{Less Dispersive Dominating Set:}~\cite{JRLP}
A set $D$ is said to be a less dispersive dominating set if the set $D$ is dominating, connected and less dispersive. The cardinality of minimum  less dispersive dominating set of $G$ is denoted by $\gamma_{comf}(G)$. A set of vertices is said to be a  $\left\langle \gamma_{comf}-set \right\rangle$, if it is a less dispersive dominating set with cardinality  $\gamma_{comf}(G)$. 
\end{definition}
\begin{definition}\textbf{Comfortable Team:}~\cite{JRLP}
A team $\left\langle D \right\rangle$ is said to be a comfortable team if $\left\langle D \right\rangle$  is  less dispersive and dominating. Minimum comfortable team is a comfortable team with the condition:  $|D|$ is minimum.
\end{definition}
\textbf{Example 1:} Consider the graph (network) $G$ in Figure~\ref{fig1}.  
\begin{figure}[h]
%\footnotesize\centering
\centerline{\includegraphics[width=3in]{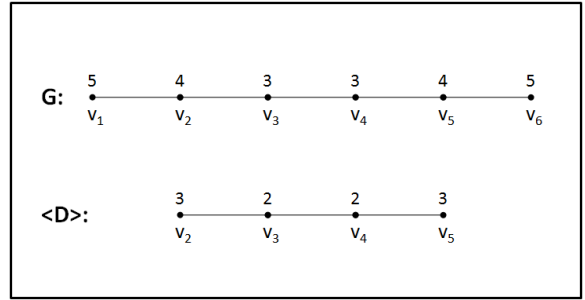}}
\caption{A Network and its Comfortable Team}
\label{fig1}
\end{figure}

Here, $G$ is a path of length six ($P_6$). $D= \{v_2,v_3,v_4,v_5\}$. The induced sub graph  $\left\langle D \right\rangle$ of $G$ forms a path of length four ($P_4$) and so it dominates all the vertices in $V-D$. Also, $\left\langle D \right\rangle$  forms the comfortable team of $G$,  because \\$e_{\left\langle D\right\rangle}(v_2) = 3 < 4 =   e_G(v_2). \  \Rightarrow e_{\left\langle D\right\rangle}(v_2) <  e_G(v_2)$. Similarly, $e_{\left\langle D\right\rangle}(v_i) <  e_G(v_i)$ for every $i = 3,4,5$. Thus, $D$ forms less dispersive set and  hence $\left\langle D \right\rangle$  forms the comfortable team of $G$. $\Rightarrow \gamma_{comf}(P_6) = 4$.

So, the problem is coined as: Find a team which is dominating, connected and less dispersive. 
It is  to be noted that there are many graphs  which do not have $\left\langle \gamma_{comf}-set \right\rangle$. So, we must try to avoid such kind of networks for successful team work.\\
\textbf{Example 2:} Consider the graph $G$ in Figure~\ref{fig2}.
\begin{figure}[h]
%\footnotesize\centering
\centerline{\includegraphics[width=3in]{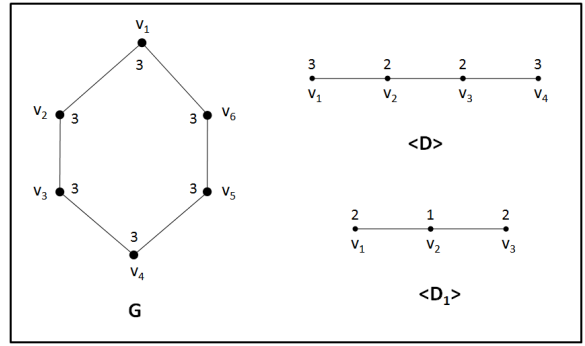}}
\caption{A Network and its CC Team}
\label{fig2}
\end{figure}

 Here, $G$ is a cycle of length six ($C_6$).
The vertices $v_1$ and $v_4$ dominate all the vertices of $G$. So, with connectedness, we can take $D=\{v_1, v_2, v_3, v_4\}$. The set $D$ dominates $G$, but $D$ is not less dispersive, because,\\
$e_{\left\langle D\right\rangle}(v_1) = 3 =  e_G(v_1)$ and  $e_{\left\langle D\right\rangle}(v_4) = 3 =  e_G(v_4)$. The vertices $v_1$ and $v_4$  maintained the original  eccentricity as in $G$. Thus, $e_{\left\langle D\right\rangle}(v) <  e_G(v)$, for every vertex $V \in D$ is \textbf{not satisfied}. So, $D$ is not less dispersive and hence $\left\langle D\right\rangle$ is not a comfortable team. \\
Also, $D_1 =\{v_1, v_2, v_3\}$ forms less dispersive set in $G$, (from Figure~\ref{fig2}), but $D_1$ is not dominating. The vertex $v_5$ is left undominated. 

\begin{note}
\label{imp}
From the above discussion in Example 2, we get,
 \begin{itemize}
	\item  the less dispersive set may not be dominating 
\item  the dominating set may not be less dispersive.
\end{itemize} 

So, under one of these two cases, the graph $G$ does not possess comfortable team. It is to be noted that not only $C_6$, but all the cycles $C_n$, do not possess a comfortable team. Also, there are infinite families of graphs which do not possess comfortable team.
\end{note}

\textbf{Disadvantage of the comfortable team~\cite{JRLP}:}\\ 
As discussed in the Example 2, comfortable team does not exist in any given network. Infinite families of networks do not possess comfortable team. 
 
The main aim of this paper is to find a team which is \textbf{more comfortable and less dispersive}, in \textbf{any given social network}.\\
So, we define a core comfortable team with a  modification in the comfortable team in the next section.

\section{Core Comfortable Team}
\label{gcomf}

As discussed in Note~\ref{imp}, the less dispersive set may not be dominating. So, we define  less dispersive set with $k^{*}-$domination.
\begin{definition}\textbf{Less Dispersive $k^{*}-$ Dominating Set:}

A set $D$ is said to be a \lq  less dispersive $k^{*}$-dominating\rq \  set if $D$ is a less dispersive set and a $k^{*}-$  dominating set. That is, 
\begin{enumerate}
\item $e_{\left\langle D\right\rangle}(v) <  e_G(v)$, for every vertex $v \in D$ (less dispersive)  
		\item $dist(D, V-D) \leq k^{*}$ ($k^{*}-$ domination). This implies that $k^{*} \leq diam(G)$, (because any vertex in $V-D$ can reach $D$ at a distance of at most $diam(G)$). 
 
\end{enumerate}
Minimum cardinality of a \lq  less dispersive $k^{*}$-dominating\rq \  set of $G$ is denoted by $\gamma_{k^{*}comf}(G)$ and maximum cardinality of a \lq  less dispersive $k^{*}$-dominating\rq \  set of $G$ is denoted by $\Gamma_{k^{*}comf}(G)$.
\end{definition}

\begin{definition}\textbf{Core Comfortable Team:}
A team $\left\langle D \right\rangle$ is said to be a  Core Comfortable (CC) team if  $\left\langle D \right\rangle$  is  less dispersive $k^{*}$-dominating. 
\begin{itemize}
\item Min CC team is a CC team with the condition: $|D|$ and $k^{*}$ are minimum 
\item Max CC team is a CC team with the condition: $|D|$ is maximum and $k^{*}$ is minimum.
\end{itemize}
\end{definition}
\textbf{Example 3:} Consider the graph $G$ ($C_6$) in Figure~\ref{fig2}. In $C_6$, $D_1= \{v_1, v_2, v_3\}$ forms a  less dispersive 2-dominating set, because
\begin{enumerate}
	\item  $e_{\left\langle D_1\right\rangle}(v_i) <  e_G(v_i)$, for $i = 1,2,3$.
	
	\item $k^{*} = 2$, because $v_5$ is reachable from $D_1$ by  distance  two, $v_4$ and $v_6$ are reachable from $D_1$ by distance one. $\Rightarrow dist(D_1, V-D_1) \leq 2$. 
\end{enumerate}
Thus, CC team exists in $C_6$ and hence $\gamma_{2comf}(G) = \Gamma_{2comf}(G) = 3$.
\begin{theorem}
\label{better NPC}
Forming core comfortable team in a given network is NP-complete.
\end{theorem}
\begin{proof}
Let $D$ be a minimum less dispersive $k^{*}-$  dominating set of $G$.\\
$\Rightarrow D$  is a connected $k^{*}-$ dominating set of $G$ (since any less dispersive set is a connected set).\\
$\Rightarrow D$  is a connected dominating set of $G^{k^{*}}$ (by definition of the graph $G^{k}$). \\
Finding $\gamma_c(G)$, for any graph $G$ is NP-complete (by Slater et al.~\cite{Slater}).\\
$\Rightarrow$ Finding $\gamma_c(G^{k^{*}})$ is NP-complete.\\
$\Rightarrow$ Finding minimum less dispersive $k^{*}-$  dominating set of $G$ is NP-complete (by above points). \\
Thus forming core comfortable team in a given network is NP-complete.
\end{proof}
Next, we give two polynomial time approximation algorithms for finding core comfortable team in any given network.
\section{Approximation Algorithm 1}
\label{algo}
In this section, we give a polynomial-time approximation algorithm for finding CC  team from a given network.
\subsection{Notation 2}
\begin{itemize}
	\item $D \rightarrow $ minimum less dispersive, $k^{*}$-dominating set.
	\item $D_1\rightarrow $  minimal  less dispersive, $k$-dominating set, (output of our algorithm). $\Rightarrow |D_1| \geq |D|$.
	\item $k^{*} \rightarrow$ the distance between two sets $D$ and $V-D$, that is, \\ $dist(D,V-D) \leq k^{*}$.
	\item $k \rightarrow$ the distance between two sets $D_1$ and $V-D_1$, that is, \\ $dist(D_1,V-D_1) \leq k$.
	\item Instant $D\rightarrow$ The set $D$ at a particular iteration.
	\item Performance ratio = $\displaystyle \frac{|minimal\ set|} {|minimum\ set|}$.
\end{itemize}
\subsection{Algorithm GOCOM} 
A polynomial time approximation algorithm for finding CC team is given below.\\
\textbf{Input:} $G$.\\
\textbf{Output:} $D_1$, which is a less dispersive, $k$-dominating set, so that $\left\langle D_1\right\rangle$ is a core comfortable team.\\
GOCOM(G)
\begin{enumerate} 
\item Choose a central vertex $v$ (ties can be broken arbitrarily) and add it to $D_1$.
\item If $diam(G)$ is even, then choose all the vertices in $N_j(v)$, for $j \leq \displaystyle \frac{diam(G)}{2}$ and add them to $D_1$.\\ 
else choose all the vertices in $N_j(v)$, for $j \leq \displaystyle \frac{(diam(G)-1)}{2}$ and add them to $D_1$.
\item Put $i = \left\lfloor \displaystyle \frac{diam(G)}{2} \right\rfloor$.
	\item  If $e_{\left\langle D_1\right\rangle}(v) < e_G(v)$, for every vertex $v \in D_1$, then Goto next step (step 5), else Goto Step 7.
 \item Put $i=i+1$.
\item Choose all vertices from $N_i(v)$ and add it to $D_1$. Then GOTO step 4.
\item Remove suitably some vertices from $D_1$ (say from $D_1 \cap N_{i-1}(v)$, from $D_1 \cap N_{i-2}(v)$, and so on) such that the condition in step 4 is satisfied.
\item Print $D_1$.
\item Stop.	
\end{enumerate} 

\begin{note}
\label{correct1}
At each iteration after forming $D_1$, we check up the condition: 
\begin{equation}
\label{ecc}
e_{\left\langle D_1\right\rangle}(v) < e_G(v),\ for\  every\  vertex\  v \in D_1.
\end{equation}
 If the condition~\ref{ecc} is not satisfied in Step 4, then the Step 7 is executed in the algorithm.  We remove some vertices from $D_1$ until the condition~\ref{ecc} is satisfied. 
 
 If the condition~\ref{ecc} is satisfied, then we add some vertices to $D_1$. There may be a question: why should the process be continued?  We add some vertices to $D_1$, in order to minimize $k$. Our aim is to minimize $k$ as well as to satisfy the  condition~\ref{ecc}.
 
 So, in Step 6, we add some vertices to $D_1$ and check up the condition~\ref{ecc}. 
 
 If the condition~\ref{ecc} is satisifed, then we proceed to add vertices to $D_1$. But, we can not go on adding vertices to $D_1$, because at one stage, the condition ~\ref{ecc} will not be satisfied. Then the step 7 will be executed.   
 
 After Step 7 is executed, the condition ~\ref{ecc} is satisfied. So, no further addition and deletion of vertices are done. The algorithm prints $D_1$ and ends.
\end{note}
\begin{note}
The algorithm GOCOM finds a CC team. If Min CC team is needed, some vertices can be removed from the output set such that $k^{*}$ is minimum. If Max CC team is needed, then some vertices could be added to the output such that condition~\ref{ecc} is satisified and $k^{*}$ is minimum.
\end{note}
\subsection{Illustration}
\label{illus}
Consider the network (graph) $G$ as in the Figure~\ref{fig3}. In this network, $diam(G) = 5$. 
\begin{figure}[h]
%\footnotesize\centering
\centerline{\includegraphics[width=3in]{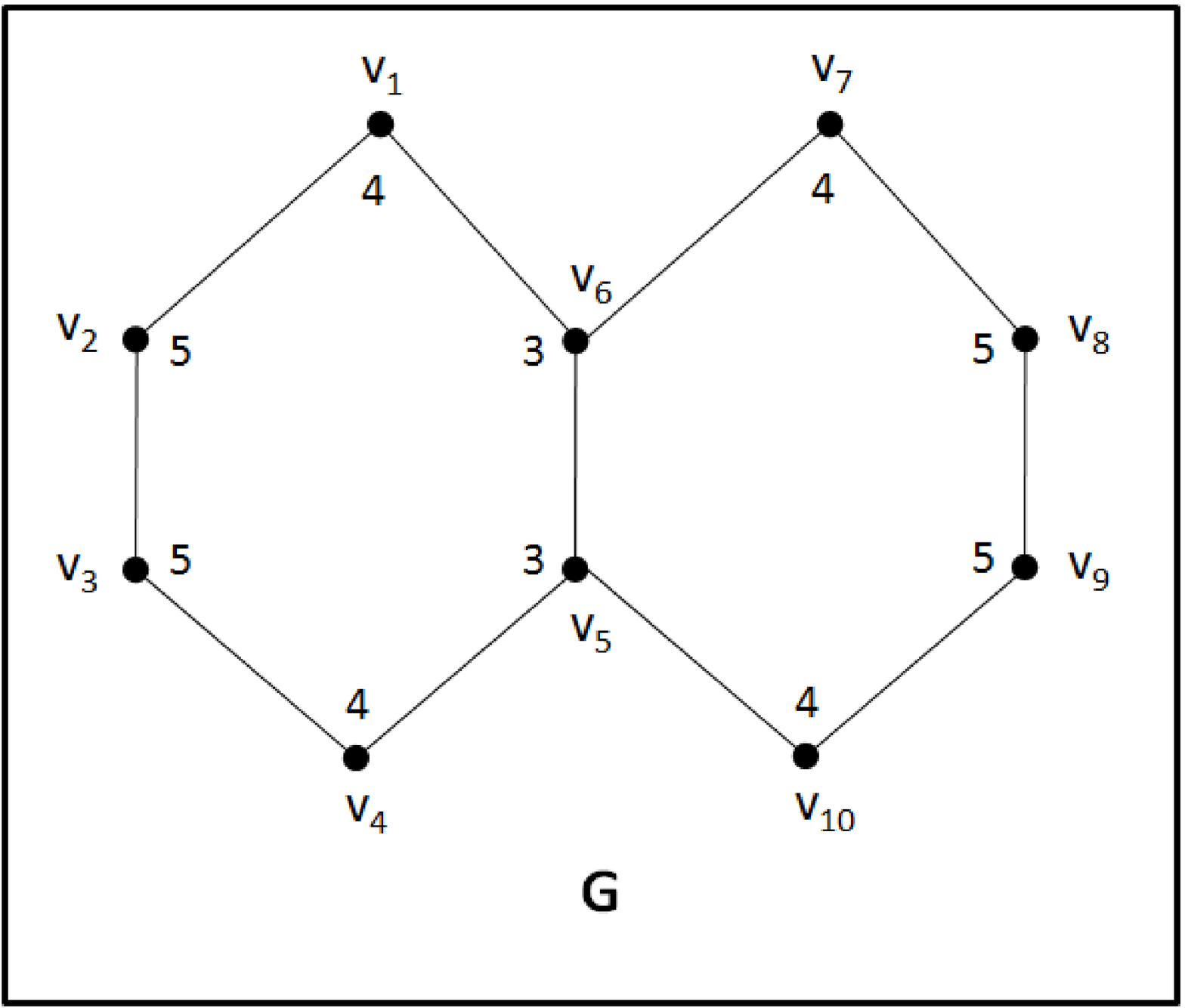}}
\caption{Illustration Figure}
\label{fig3}
\end{figure}

Inititally, let us choose the central vertex $v_6$ and add them to $D_1$. As $diam(G)$ is odd, let us choose up to $\frac{diam(G)-1}{2} = 2$ neighborhoods of $v_6$ and add them to $D_1$. Now, instant $D_1 = \{v_1, v_2, v_4, v_5, v_6, v_7, v_8, v_{10}\}$. But, we can see in Figure~\ref{fig4} that vertices $v_4, v_5$ and $v_{10}$ of $D_1$ at an intermediate interation violate the condition~\ref{ecc}. In order to make $D_1$ maintain the condition~\ref{ecc}, we remove some vertices from $D_1$.

\begin{figure}[h]
%\footnotesize\centering 
\centerline{\includegraphics[width=3in]{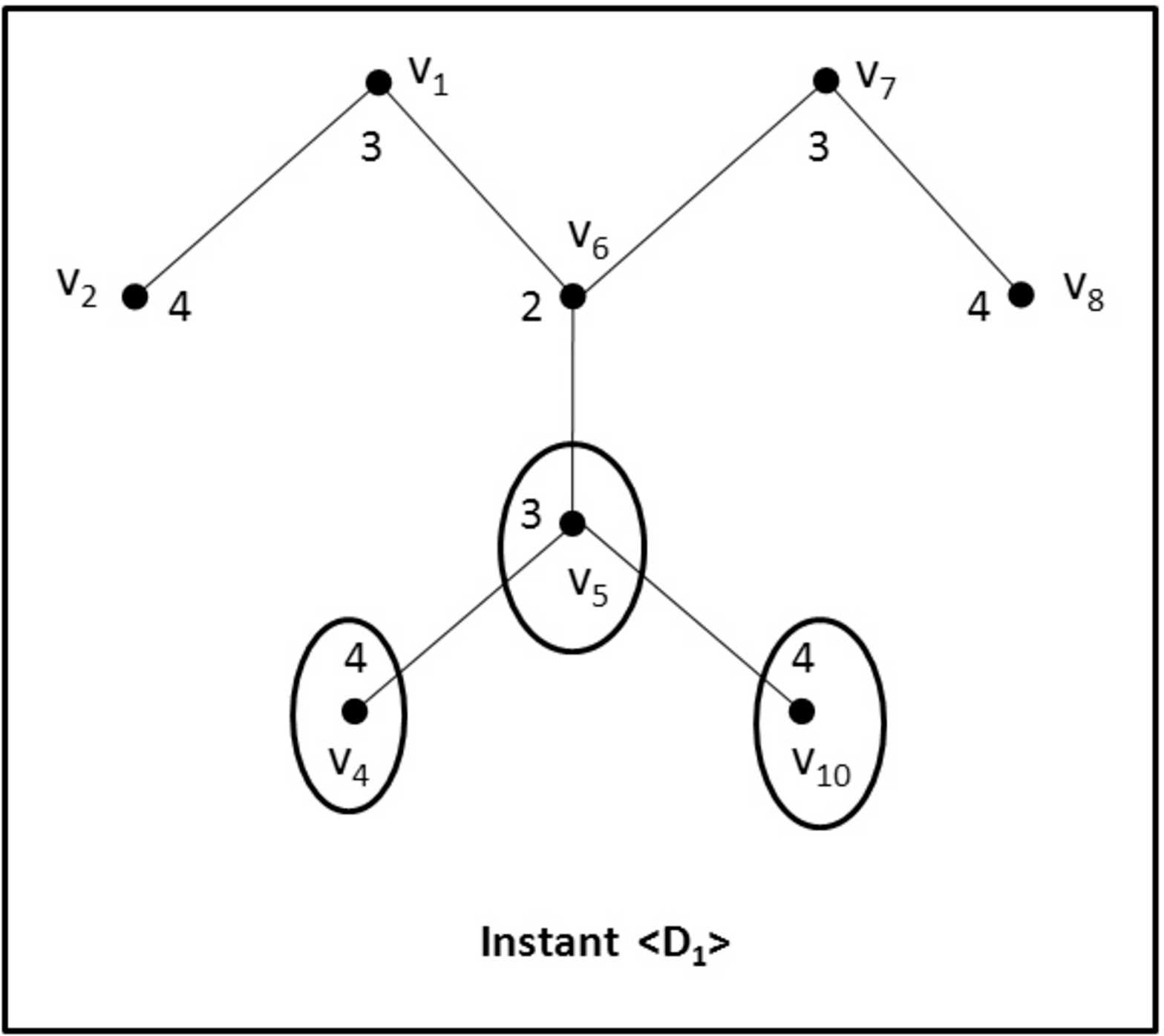}}
\caption{Intermediate Iteration showing vertices violating the condition (1)}
\label{fig4}
\end{figure}

We can either remove those vertices $v_4, v_5$ and $v_{10}$ or the vertices $v_2$ and $v_8$, so that the condition~\ref{ecc} is satisfied. So, we get two different outputs. Let us make the first one as $D_1 = \{v_1, v_2, v_6, v_7, v_8\}$ and the other one as $D_2 = \{v_1, v_4, v_5, v_6, v_7, v_{10}\}$. Refer Figure~\ref{fig5}. Both outputs satisfy the condition~\ref{ecc} and hence both $\left\langle D_1 \right\rangle$ and $\left\langle D_2 \right\rangle$ are core comfortable teams. But, for$D_1$, we get $k = 2$ and for $D_2$, we get $k=1$.

$D_1$ is a maximal CC team for $k=2$ and hence $\Gamma_{2comf}(G) = 5$. Also, for $k=2$, we see that $D_3 = \{v_5, v_6\}$ forms a minimal  CC team. The vertices $v_5$ and $v_6$ are suuficient and every vertex in $V-D_3$ is reachable from $D_3$ by a distance of at most two. Also, $D_3$ satisfies the condition~\ref{ecc}. Refer Figure~\ref{fig5}. Thus, $D_3$ is a minimal CC team  and hence $\gamma_{2comf}(G) = 2$.\\

For $k=1$, $D_2$ is a minimal as well as a maximal CC team  and hence $\gamma_{1comf}(G)= \Gamma_{1comf}(G) = 6$.  

\begin{figure}[h]
%\footnotesize\centering
\centerline{\includegraphics[width=3in]{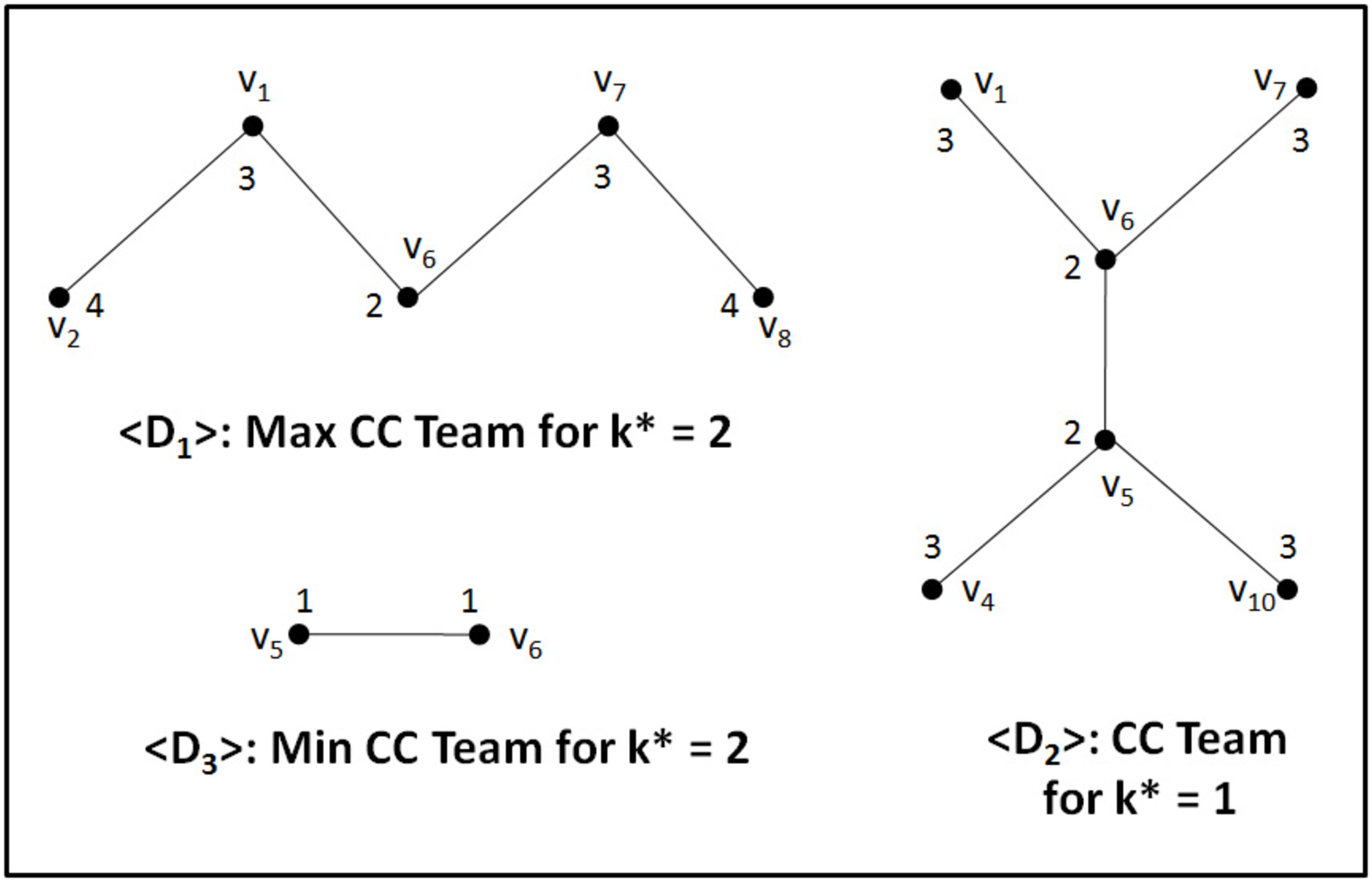}}
\caption{Three Different Outputs}
\label{fig5}
\end{figure}

\begin{note}
\label{unique}
From the Illustration~\ref{illus}, we can observe that the core comfortable team is not unique for a given network. A social network may have many  core comfortable teams. We can choose one team among all the teams whichever is suitable for a particular situation.  We can choose any CC team (minimum or maximum) according to our need for a particular situation.
\end{note}

\subsection{Time Complexity of the Algorithm GOCOM}
\label{time}
Let us discuss the time complexity of the algorithm as follows:
 
The definition of CC team and the algorithm GOCOM is dependent on eccentricity of every vertex. So, we have to find eccentricity of every vertex of $G$. By Performing Breadth-First Search (BFS) method from each vertex, one can determine the distance from each vertex to every other vertex. The worst case time complexity of BFS method for one vertex is $O(n^{2})$. As the BFS is method is done for each vertex of $G$, the resulting algorithm has worst case time complexity $O(n^{3})$. As eccentricity of a vertex $v$ is defined as $e(v) = \max\{d(u,v): u\in V(G)\}$, finding eccentricity of vertices of $G$ takes at most $O(n^{3})$.

Thus, the total worst case time complexity of the algorithm is at most $O(n^{3})$.
\subsection{Correctness of the Algorithm GOCOM}
\label{correct}
In order to prove that the  algorithm GOCOM yields a CC team, it is sufficient to prove that the condition in Step 4 of the algorithm is satisfied. As proof follows from Note~\ref{correct1}, we state the following theorem without proof.
\begin{theorem}
$e_{\left\langle D_1\right\rangle}(v) < e_G(v)$, for every vertex $v \in D_1$, where $D_1$ is the output of our Algorithm GOCOM. 
\end{theorem}
 
\subsection{Performance Ratio of the Algorithm GOCOM}
\label{ratio}
It is to be noted that the algorithm has two parameters, namely, $|D|$ and $k^{*}$. The set $D$ represents the team members of CC team,  and $k^{*}$ represents the distance between the team members and non-team members. We find performance ratio of the algorithm for finding \textbf{Min CC team}. So, it is necessary that both $D$ and $k^{*}$ should be minimized simultaneously.

The following theorem gives the performance ratio for finding the minimum less dispersive $k^{*}-$ dominating set.
Performance ratio of the algorithm for finding the minimum CC team is  equal to  $|D_1|/|D|$.
\begin{theorem}
\label{per}
The performance ratio of the algorithm for finding Min CC team  is at most $O(\ln \Delta(G))$.
\end{theorem}
\begin{proof}
Let $D_1$ be a minimal less dispersive $k-$ dominating set of $G$ and \\
let $D$ be a minimum less dispersive $k^{*}-$ dominating set of $G$.\\
$\Rightarrow D$  is a connected $k^{*}-$ dominating set of $G$ (since any less dispersive set is a connected set).\\
$\Rightarrow D$  is a connected dominating set of $G^{k^{*}}$ (by definition of the graph $G^{k^{*}}$). \\
Performance ratio for finding $\gamma_c(G)$ is at most $O(\ln \Delta(G))$.\\
$\Rightarrow$ Performance ratio for finding $\gamma_c(G^{k^{*}})$ is at most $O(\ln \Delta(G^{k^{*}}))$.\\
But, $\Delta(G^{k^{*}}) \leq (\Delta(G)) ^ {k^{*}}$.\\
$\Rightarrow$ Performance ratio for finding $\gamma_c(G^{k^{*}})$ is at most $O(\ln (\Delta(G)) ^ {k^{*}}) = O(k^{*} \ln \Delta(G)) = O(\ln \Delta(G))$.\\
Thus, the performance ratio of the algorithm for finding Min CC team  is at most $O(\ln \Delta(G))$.
\end{proof}

Next, let us give the performance ratio for finding $k^{*}$.
Performance ratio of the algorithm for finding $k^{*}$ is equal to $k / k^{*}$.
\begin{theorem}
\label{PRgo}
The performance ratio of the Algorithm GOCOM for finding $k^{*}$ is at most $r(G)$.
\end{theorem}
\begin{proof}
Let us recall from Notation 2:  $dist(D,V-D) \leq k^{*}$ and  $dist(D_1,V-D_1) \leq k$, that is, $k$ and $k^{*}$ denote the minimal parameter  and the minimum parameter respectively.

In the Algorithm GOCOM, we start from a central vertex. So, the output set $D_1$ always contains at least one central vertex. This implies that any vertex in $V-D_1$ is reachable from  $D_1$ by distance at most $r(G)$. Thus, $k \leq r(G)$.

Also, as 1-dominating set is possible in many cases, $k^{*} \geq 1$.\\
Thus, $\frac{k}{k^{*}} \leq r(G)$.
\end{proof}
\section{Approximation Algorithm 2}
\label{second}
In this section, we give another approximation algorithm for finding core comfortable team and analyze its time complexity, correctness and performance ratio.\\
Algorithm CONCOMF:\\
Input: $G$.\\
Output: $D_1$, which is a less dispersive $k-$ dominating set, so that $\left\langle D_1\right\rangle$ is a CC team.\\
CONCOMF($G$)
\begin{enumerate}
\item Find a connected $k-$ dominating set of $G$ (using an  approximation algorithm).
\item Store  all the  vertices in $D_1$. 
\item  If $e_{\left\langle D_1\right\rangle}(v) < e_G(v)$, for every vertex $v \in D_1$, then GOTO Step 5 else Goto Step 4 (next step).
 \item Remove suitably some vertices from $D_1$  such that the condition in step 3 is satisfied .
\item Print $D_1$.
\item Stop.	
\end{enumerate} 
\begin{note}
The algorithm CONCOMF finds a CC team. If Min CC team is needed, in Step 1, find a minimal $k-$ dominating set of $G$. If Max CC team is needed, in Step 1, find a maximal connected $k-$ dominating set of $G$. Executing the Algorithm CONCOMF in graph $G$ of Figure~\ref{fig3}, we get $D_1$ is maximal CC team and $D_3$ is minimal CC team for $k=2$ and $D_2$ is a minimal and maximal CC team for $k=1$. 
\end{note}

\subsection{Time Complexity of the Algorithm CONCOMF}
The worst case time complexity of the approximation algorithm for finding connected dominating set is at most $O(n^{3})$. Also, as discussed in the Section~\ref{time}, the definition of  CC team is dependent on eccentricity of every vertex and  finding eccentricity of vertices of $G$ takes at most $O(n^{3})$ in worst case. 
 
Thus, the total worst case time complexity of the algorithm is at most $O(n^{3})$.
\subsection{Correctness of the Algorithm CONCOMF}
At the end of Step 4 in Algorithm CONCOMF, the output $D_1$ satisfies the condition~\ref{ecc}. So, we state the following theorem without proof.
\begin{theorem}
$e_{\left\langle D_1\right\rangle}(v) < e_G(v)$, for every vertex $v \in D_1$, where $D_1$ is the output of our Algorithm GOCOM. 
\end{theorem}

\subsection{Performance Ratio Of the Algorithm CONCOMF}
\label{ratio1}
As discussed in the Section~\ref{ratio}, for finding a Min CC team, it is necessary that the two parameters $|D|$ and $k^{*}$ should be minimized simultaneously.

The following theorem gives the performance ratio for finding the minimum less dispersive $k^{*}-$ dominating set. As proof follows form Theorem~\ref{per}, we state the following theorem without proof.

\begin{theorem}
\label{percon}
The performance ratio of the Algorithm CONCOMF for finding  Min CC team  is at most $O(\ln \Delta(G))$.
\end{theorem}
Next, we give performance ratio of the algorithm CONCOMF for finding $k^{*}$.
\begin{theorem}
\label{PR}
The performance ratio of the Algorithm CONCOMF for finding $k^{*}$ is at most $diam(G)$.
\end{theorem}
\begin{proof}
As 1-dominating set is possible in many cases, $k^{*} \geq 1$.

Also, any vertex in $D_1$ is reachable form $V-D_1$ by a distance of at most $diam(G)$. This implies that $k \leq diam(G)$.

Thus, $\frac{k}{k^{*}} \leq diam(G)$. 
\end{proof}
\section{Advantages}
\label{adv}
\textbf{Advantage 1:} The algorithms give \textbf{good results in scale free networks} for finding core comfortable team.\\
\textbf{Explanation:} The performance ratio of the Algorithms GOCOM and CONCOMF for finding $k^{*}$, are dependent on $r(G)$ and $diam(G)$ respectively (by Theorems~\ref{PR} and ~\ref{PRgo}). It is known that  the growing scale-free networks have almost constant diameter in practice. So, the algorithms give constant performance ratio in scale free networks.\\  
\textbf{Advantage 2:} The algorithms can be applied in any \textbf{random networks} for finding CC team.\\
\textbf{Explanation:}  From the theorems~\ref{per}, \ref{PRgo}, \ref{percon} and~\ref{PR}, it is clear that the performance ratio of the algorithm for finding $|D|$ and $k^{*}$ is dependent on  $diam(G)$ and $\Delta(G)$. As both these terms can be expressed in terms of the probability $p$, the performance ratio of the algorithm can be easily obtained for random networks in terms of $p$.\\
\textbf{Advantage 3:} CC team can be obtained in \textbf{disconnected networks} also using the algorithms. \\
\textbf{Explanation:} If the network (graph) is disconnected, then as mentioned in the Section~\ref{intro}, algorithm can can be applied to each connected component of the network.  Thus, algorithm can be applied to find CC team in any given network. 

\section{Conclusion}
\label{conc}
In this paper,  core comfortable team of a social network is defined. It is proved that forming core comfortable  team in any given network is NP-complete. Two polynomial time approximation algorithms are given for finding a CC team in any given network and the time complexity of those algorithms are given to be $O(n^{3})$, where $n$ is the number of vertices of $G$. The correctness of the algorithms are analyzed.  The performance ratio of the algorithms for finding Min CC team is proved to be $O(\ln \Delta)$, where $\Delta$ is the maximum degree of $G$. The performance ratio of the two algorithms for finding $k^{*}$ is proved to be at most $r(G)$ and $diam(G)$ respectively. It is also proved that the algorithms give good results in scale free networks.
\subsection{Future Work}
The algorithms can be applied in a particular social network, for example, Poisson network and can be tried to reduce the performance ratio in that network. Algorithms can be implemented to get exact values also in some particular networks, for example scale free networks and so on. Also, as discussed in Note~\ref{unique}, a social network can have many CC teams. So, we can analyse the different situations for which the  CC teams are suitable.

%\begin{acknowledgements}
%If you'd like to thank anyone, place your comments here
%and remove the percent signs.
%\end{acknowledgements}

% BibTeX users please use one of
%\bibliographystyle{spbasic}      % basic style, author-year citations
%\bibliographystyle{spmpsci}      % mathematics and physical sciences
%\bibliographystyle{spphys}       % APS-like style for physics
%\bibliography{}   % name your BibTeX data base
%\bibliographystyle{elsarticle-num}
%\bibliography{<your-bib-database>}

% Non-BibTeX users please use

\end{document}